\documentclass[envcountsect,runningheads,a4paper,leqno]{llncs}

\usepackage{subcaption} 
\usepackage[utf8]{inputenc}
\usepackage[T1]{fontenc}
\usepackage{amsmath}
\usepackage{amssymb}
\usepackage{todonotes}
\usepackage{csquotes}
\usepackage{enumitem}
\usepackage{placeins}
\usepackage{thm-restate}
\usepackage{mathtools}
\usepackage{multirow}
\usepackage{booktabs}
\usepackage[colorlinks]{hyperref}
\usepackage[capitalize]{cleveref}

\DeclareRobustCommand*{\leftmodels}{%
  \Relbar\joinrel\mathrel{|}%
}

\usetikzlibrary{automata,positioning}

\tikzset{state/.style={
  rectangle,
  rounded corners,
  draw=black,
  minimum height=2em,
  minimum width=2em,
  align=center}
}

\tikzset{every picture/.append style={initial text=}}
\tikzset{accepting/.style = {double}}
\tikzset{>=stealth}
\tikzset{parallel above/.append style={transform canvas={yshift= 1mm}}}
\tikzset{parallel below/.append style={transform canvas={yshift=-1mm}}}
\tikzset{parallel right/.append style={transform canvas={xshift= 1.3mm}}}
\tikzset{parallel left/.append style={transform canvas={xshift=-1.3mm}}}

\newcommand{\lang}[0]{\ensuremath{\mathcal{L}}}

\newcommand{\true}{{\ensuremath{\mathbf{t\hspace{-0.5pt}t}}}} 
\newcommand{\false}{{\ensuremath{\mathbf{ff}}}}
\newcommand{\F}{{\ensuremath{\mathbf{F}}}}
\newcommand{\G}{{\ensuremath{\mathbf{G}}}}
\newcommand{\X}{{\ensuremath{\mathbf{X}}}}
\newcommand{\U}{{\ensuremath{\mathbf{U}}}}
\newcommand{\W}{{\ensuremath{\mathbf{W}}}}

\newcommand{\M}{{\ensuremath{\mathbf{M}}}}
\newcommand{\R}{{\ensuremath{\mathbf{R}}}}
\newcommand{\GF}{\ensuremath{\mathbf{G\hspace{-0.07cm}F\!}\,}}
\newcommand{\FG}{\ensuremath{\mathbf{F\hspace{-0.07cm}G\!}\,}}

\newcommand{\sfmu}{{\ensuremath{\mathbb{\mu}}}}

\newcommand{\sfnu}{{\ensuremath{\mathbb{\nu}}}}

\newcommand\N{\ensuremath{\mathbb{N}}}

\newcommand{\hole}{[ \quad ]}
\newcommand{\rank}[1]{\mathit{rank}({#1})}

\spnewtheorem{factt}{Fact}[section]{\bfseries}{\itshape}

\newcommand{\fnf}{1-form}
\newcommand{\snf}{1-2-form}

\newcommand{\nnodes}[1]{|#1|}
\newcommand{\gfba}[1]{n_{\text{lim}}({#1})}
\newcommand{\ubw}[1]{n_{u}({#1})}

\begin{document}


\title{A Simple Rewrite System for \\ the Normalization of Linear Temporal Logic\thanks{This work was partially supported by the Deutsche Forschungsgemeinschaft (DFG) under projects 183790222, 317422601, and 436811179; by the European Research Council (ERC) under the European Union's Horizon 2020 research and innovation programme under grant agreement No~787367 (PaVeS); by the Spanish MCI project ProCode (PID2019-108528RB-C22); and by the Spanish MU grants FPU17/02319 and EST21/00536.}}

\author{Javier Esparza\inst{1}, Rub\'en Rubio \inst{2}, and Salomon Sickert \inst{3}}
\authorrunning{J. Esparza, R. Rubio, S. Sickert}
\institute{Technical University of Munich,  Germany 
\and
Universidad Complutense de Madrid, Spain
\and
The Hebrew University, Jerusalem, Israel
}
\maketitle              

\begin{abstract}
In the mid 80s, Lichtenstein, Pnueli, and Zuck showed that every formula of Past LTL (the extension of Linear Temporal Logic with past operators) is equivalent to a conjunction of formulas of the form $\G\F \varphi \vee \F\G \psi$, where $\varphi$ and $\psi$ contain only past operators. Some years later, Chang, Manna, and Pnueli derived a similar normal form for LTL. Both normalization procedures have a non-elementary worst-case blow-up, and follow an involved path from formulas to counter-free automata to star-free regular expressions and back to formulas. In 2020, Sickert and Esparza presented a direct and purely syntactic normalization procedure for LTL yielding a normal form similar to the one by Chang, Manna, and Pnueli, with a single exponential blow-up, and applied it to the problem of constructing a succinct deterministic $\omega$-automaton for a given formula. However, their procedure had exponential time complexity in the best case. In particular, it does not perform better for formulas that are almost in normal form. In this paper we present an alternative normalization procedure based on a simple set of rewrite rules.
\end{abstract}

\section{Introduction}

In the late 1970s, Amir Pnueli introduced Linear Temporal Logic (LTL) into
computer science as a framework for specifying and verifying concurrent programs \cite{Pnueli77,Pnueli81}, 
a contribution that earned him the 1996 Turing Award.  During the 1980s and the early 1990s, Pnueli proceeded to study the properties expressible in LTL in collaboration with other researchers. In 1985, Lichtenstein, Pnueli and Zuck introduced a classification of LTL properties \cite{LPZ85}, later described in detail by Manna and Pnueli, who called it the \emph{safety-progress} hierarchy in \cite{MannaP89,MannaPnueli91}. These works consider an extended version of LTL with past operators, called Past LTL. The safety-progress hierarchy consists of a \emph{safety} class of formulas, and five \emph{progress} classes. The classes are defined semantically in terms of their models, and the largest class, called the \emph{reactivity} class in \cite{MannaP89,MannaPnueli91}, contains all properties expressible in LTL.  Manna and Pnueli provide syntactic characterizations of each class. In particular, they prove a fundamental theorem showing that every reactivity property is expressible as a conjunction of formulas of the form $\G\F\varphi \vee \F\G\psi$, where $\F \chi$ and $\G \chi$ mean that $\chi$ holds at some and at every point in the future, respectively, and  $\varphi, \psi$ only contain past operators. 

In 1992, Chang, Manna, and Pnueli presented a different and very elegant characterization of the safety-progress hierarchy in terms of standard LTL without past operators, containing only the future operators $\X$ (next), $\U$ (until), and $\W$ (weak until) \cite{ChangMP92}.  They show that every reactivity formula is equivalent to an LTL formula in negation normal form, such that every path through the syntax tree contains at most one alternation of $\U$ and $\W$. We call this fundamental result the Normalization Theorem. In the notation of \cite{CernaP03,PelanekS05,SickertE20}, which mimics the definition of the $\Sigma_i$, $\Pi_i$, and $\Delta_i$ classes of the arithmetical and polynomial hierarchies, they proved that every LTL formula is equivalent to a $\Delta_2$-formula.

While these normal forms have had large conceptual impact in model checking, automatic synthesis, and deductive verification (see e.g. \cite{PitermanP18} for a recent survey), the normalization \emph{procedures} have had none. In particular, contrary to the case of propositional or first-order logic, they have not been implemented in tools. The reason is that they are not direct, have high complexity, and their correctness proofs are involved. The proof of the Normalization Theorem sketched in \cite{ChangMP92} (to the best of our knowledge, a full proof was never published) relies on the 1985 theorem by Lichtenstein, Pnueli and Zuck, a complete proof of which can be found in Zuck's PhD Thesis \cite{Zuck86}. Zuck's proof translates the initial Past LTL formula into a counter-free semi-automaton, then applies the Krohn-Rhodes decomposition and other results to translate the automaton into a star-free regular expression, and finally translates this expression into a reactivity formula with a non-elementary blow-up. It is remarkable that, despite this prominence, only little progress has been made to improve Zuck's non-elementary normalization procedure, even though no lower bound was known. 

On the one hand, Maler and Pneuli have presented a double-exponential\footnote{For further details we refer the reader to \cite[Remark 1]{DBLP:conf/fossacs/BokerLS22}.} construction, based on the Krohn-Rhodes decomposition, translating a deterministic counter-free automaton into a Past LTL formula \cite{MP90,MP94,Mal10}. On the other hand, building upon this work, Boker, Lehtinen, and Sickert discovered a triple-exponential construction translating into a standard LTL formula without past operators \cite{DBLP:conf/fossacs/BokerLS22}. Noticeably, both constructions yield formulas in the normal forms for Past LTL (\cite{MannaP89,MannaPnueli91}) and standard LTL (\cite{ChangMP92,CernaP03,PelanekS05,SickertE20}), respectively.

In 2020, two of us presented a novel proof of the Normalization Theorem in \cite{SickertE20} (based on  \cite{Sickert19}). We showed that every formula $\varphi$ of LTL is equivalent to a formula of the form 
\[\bigvee_{M \subseteq \sfmu(\varphi), N \subseteq \sfnu(\varphi)} \varphi_{M,N}\]
\noindent where $\sfmu(\varphi)$ and $\sfnu(\varphi)$ are the sets of subformulas of $\varphi$ with top operator in $\{\U, \M\}$ and $\{ \W, \R\}$, respectively, and $\varphi_{M,N}$ is a $\Delta_2$-formula obtained from $\varphi$, $M$, and $N$ by means of a few syntactic rewrite rules. This yields a normalization procedure with single exponential complexity, which was applied in \cite{SickertE20,Sickert19} to the problem of translating LTL formulas into deterministic and limit-deterministic $\omega$-automata.

Despite being a clear improvement on the previous indirect and non-elementary procedures, the normalization algorithm of \cite{SickertE20} still has a problem: Since it has to consider all possible sets $M$ and $N$,  it has exponential time complexity \emph{in the best case}. Moreover, the algorithm is not goal-oriented, in the sense that it does not only concentrate on those parts of the formula that do not belong to $\Delta_2$. Consider for example a family of formulas 
$$\varphi_n = ((a \U b) \W  c)  \U  \psi_n$$
\noindent where $a, b, c$ are atomic propositions and $\psi_n$ is some very large formula containing only the \W\ operator. Intuitively, $\psi_n$ does not need to be touched by a normalization procedure, the only problem lies in the alternation \U-\W-\U\ along the leftmost branch of the syntax tree. However, the procedure of \cite{SickertE20} will be exponential in the number of \W-subformulas of $\psi_n$.

In this paper we provide a normalization procedure that solves these problems. The procedure is similar to the one for bringing a Boolean formula in conjunctive normal form (CNF). Recall that a Boolean formula is in CNF if in its syntax tree no conjunctions are below disjunctions, and only atomic propositions are below negations. The rewrite rules allow us to eliminate a node that violates one of these conditions; for example, if a conjunction is below a disjunction, we distribute the conjunction over the disjunction. In the case of LTL, instead of conjunctions and disjunctions we have to deal with different kinds of temporal operators, but we can still characterize the normal form in terms of constraints of the form ``no X-node of the syntax tree is below a Y-node''. Our rewrite rules eliminate nodes violating one of these constraints. 

The paper is organized as follows. Section \ref{sec:prelims} introduces the syntax and semantics of LTL. Section \ref{sec:hierarchy} defines the Safety-Progress hierarchy, and  recalls the Normalization Theorem of Chang, Manna, and Pnueli. Section \ref{sec:main} presents the rewrite system, and proves it correct. Section \ref{sec:summary} summarizes the normalization algorithm derived from the rewrite system, and Section \ref{sec:extensions} introduces some derived results and some extensions of the algorithm. Finally, Section \ref{sec:experiments} reports on an experimental evaluation.

\section{Preliminaries}
\label{sec:prelims}

Let $\Sigma$ be a finite alphabet. A \emph{word} $w$ over $\Sigma$ is an infinite sequence of letters $a_0 a_1 a_2 \dots$ with $a_i \in \Sigma$ for all $i \geq 0$, and a language is a set of words. A \emph{finite word} is a finite sequence of letters. The set of all words (finite words) is denoted $\Sigma^\omega$ ($\Sigma^*$). We let $w[i]$ (starting at $i=0$) denote the $i$-th letter of a word $w$. The finite infix $w[i]w[i+1]\dots w[j - 1]$ is abbreviated with $w_{ij}$ and the infinite suffix $w[i] w[i+1] \dots$ with $w_{i}$. We denote the infinite repetition of a finite word $a_0 \dots a_n$ by $(a_0 \dots a_n)^\omega = a_0 \dots a_n a_0 \dots a_n a_0 \dots$. A set of (finite or infinite) words is called a language. 

\begin{definition}
\label{def:ltlsyntax}
LTL formulas over a set $Ap$ of atomic propositions are constructed by the following syntax:
\begin{align*}
\varphi \Coloneqq \; & \true \mid \false \mid a \mid \neg a \mid \varphi \wedge \varphi \mid \varphi\vee\varphi \\ 
                     & \mid \X\varphi \mid \varphi\U\varphi \mid \varphi\W\varphi \mid \varphi\R\varphi \mid \varphi\M\varphi 
\end{align*}
\noindent where $a \in Ap$ is an atomic proposition and $\X$, $\U$, $\W$, $\R$, and $\M$ 
are the next, (strong) until, weak until, (weak) release, and strong release operators, respectively. 
\end{definition}

The inclusion of both the strong and weak until operators as well as the negation normal form are essential to our approach. The operators $\R$ and $\M$, however, are only added to ensure that every formula of length $n$ in the standard syntax, with negation but only the until operator, is equivalent to a formula of length $O(n)$ in our syntax. They can be removed  at the price of an exponential blow-up when translating formulas with occurrences of $\R$ and $\M$ into formulas without. The semantics is defined as usual:

\begin{definition}
\label{def:ltlsemantics}
Let $w$ be a word over the alphabet $\Sigma := 2^{Ap}$ and let $\varphi$ be a formula. The satisfaction relation $w \models \varphi$ is inductively defined as the smallest relation satisfying:
{\arraycolsep=1.8pt%
\[\begin{array}[t]{lclclcl}
w \models \true               & &  \mbox{ for every $w$ } & &                                                   \\
w \not \models \false         & &  \mbox{ for every $w$ }                                                      \\
w \models a & \mbox{ iff }    & a \in w[0]                                                     \\
w \models \neg a              & \mbox{ iff } & a \notin w[0]                                   \\
w \models \varphi \wedge \psi & \mbox{ iff } & w \models \varphi \text{ and } w \models \psi   \\
w \models \varphi \vee \psi   & \mbox{ iff } & w \models \varphi \text{ or } w \models \psi    \\
w \models \X \varphi      & \mbox{ iff } & w_1 \models \varphi \\
w \models \varphi \U \psi & \mbox{ iff } & \exists k. \, w_k \models \psi \text{ and } \forall j < k. \, w_j \models \varphi \\
w \models \varphi \M \psi & \mbox{ iff } & \exists k. \, w_k \models \varphi \text{ and } \forall j \leq k. \, w_j \models \psi \\
w \models \varphi \R \psi & \mbox{ iff } & \forall k. \, w_k \models \psi \text{ or } w \models \varphi\M \psi \\
w \models \varphi \W \psi & \mbox{ iff } & \forall k. \, w_k \models \varphi \text{ or } w \models \varphi\U \psi
\end{array}\]}%
We let $\lang(\varphi) \coloneqq \{ w \in \Sigma^\omega : w \models \varphi\}$ denote the language of $\varphi$.
We overload the definition of $\models$ and write $\varphi \models \psi$ as a shorthand for $\lang(\varphi) \subseteq \lang(\psi)$. Two formulas $\varphi$ and $\psi$ are \emph{equivalent}, denoted $\varphi \equiv \psi$, if $\lang(\varphi) = \lang(\psi)$. Further, we use the abbreviations 
$\F \varphi \coloneqq \true \, \U \, \varphi$ (eventually) and $\G \varphi \coloneqq \false \, \R \, \varphi$ (always). 
\end{definition}

\section{The Safety-Progress Hierarchy}
\label{sec:hierarchy}

We recall the hierarchy of temporal properties studied by Manna and Pnueli \cite{MannaP89} following the formulation of {\v{C}}ern{\'{a}} and Pel{\'{a}}nek \cite{CernaP03}. The definition formalizes the intuition that e.g. a safety property is
violated by an execution  if{}f one of its finite prefixes is ``bad'' or, equivalently, satisfied by an execution if{}f all its finite prefixes belong to a language of good prefixes.

\begin{definition}[\cite{MannaP89,CernaP03}]
Let $P \subseteq \Sigma^\omega$ be a property over $\Sigma$.
\begin{itemize}
\item $P$ is a safety property if there exists a language of finite words $L \subseteq \Sigma^*$ such that $w \in P$ if{}f all finite prefixes of $w$ belong to $L$.
\item $P$ is a guarantee property if there exists a language of finite words $L \subseteq \Sigma^*$ such that $w \in P$ if{}f there exists a finite prefix of $w$ which belongs to $L$.
\item $P$ is an obligation property if it can be expressed as a positive Boolean combination of safety and guarantee properties.
\item $P$ is a recurrence property if there exists a language of finite words $L \subseteq \Sigma^*$ such that $w \in P$ if{}f infinitely many prefixes of $w$ belong to $L$.
\item $P$ is a persistence property if there exists a language of finite words $L \subseteq \Sigma^*$ such that $w \in P$ if{}f all but finitely many prefixes of $w$ belong to $L$.
\item $P$ is a reactivity property if $P$ can be expressed as a positive Boolean combination of recurrence and persistence properties.
\end{itemize}
\end{definition}


The inclusions between these classes are shown in \Cref{fig:temporal_hierarchy}. 
Chang, Manna, and Pnueli give in \cite{ChangMP92} a syntactic characterization of the classes in terms of the following fragments of LTL:

\begin{definition}[Adapted from \cite{CernaP03}]
\label{def:future_hierarchy}
We define the following classes of LTL formulas:
\begin{itemize}
	\item The class $\Sigma_0 = \Pi_0 = \Delta_0$ is the least set of formulas containing all atomic propositions and their negations, and is closed under the application of conjunction and disjunction.
	\item The class $\Sigma_{i+1}$ is the least set of formulas containing $\Pi_i$ that is closed under the application of conjunction, disjunction, and the $\X$, $\U$, and $\M$ operators.
	\item The class $\Pi_{i+1}$ is the least set of formulas containing $\Sigma_i$ that is closed under the application of conjunction, disjunction, and the $\X$, $\R$, and $\W$ operators.
	\item The class $\Delta_{i+1}$ is the least set of formulas containing $\Sigma_{i+1}$ and $\Pi_{i+1}$ that is closed under the application of conjunction and disjunction.
\end{itemize}
\end{definition}

The following is a corollary of the proof of \cite[Thm. 8]{ChangMP92}:

\begin{theorem}[Adapted from \cite{CernaP03}]\label{thm:hierarchy:correspondence}
A property that is specifiable in LTL is a guarantee (safety, obligation, persistence, recurrence, reactivity, respectively) property if and only if it is specifiable by a formula from the class $\Sigma_1$, $(\Pi_1$, $\Delta_1$, $\Sigma_2$, $\Pi_2$, $\Delta_2$, respectively$).$
\end{theorem}

\begin{figure}
\begin{subfigure}[c]{0.51\columnwidth}
  \begin{center}
  \small
	\begin{tikzpicture}[x=1cm,y=0.75cm,outer sep=2pt]

    \node (padding1) at ( 0,0.3) {};
	\node (1) at ( 0, 0) {reactivity};
	\node (2) at ( 1,-1) {recurrence};
	\node (3) at (-1,-1) {persistence};
    \node (4) at ( 0,-2) {obligation};
	\node (5) at ( 1,-3) {safety};
	\node (6) at (-1,-3) {guarantee};
	\node (padding2) at ( 0,-3.2) {};

	\path
	(2) edge[draw=none] node[sloped]{$\supset$} (1)
    (3) edge[draw=none] node[sloped]{$\subset$} (1)
    
    (4) edge[draw=none] node[sloped]{$\subset$} (2)
    (4) edge[draw=none] node[sloped]{$\supset$} (3)
    
    (5) edge[draw=none] node[sloped]{$\supset$} (4)
    (6) edge[draw=none] node[sloped]{$\subset$} (4);

	\end{tikzpicture}
  \end{center}
\subcaption{Safety-progress hierarchy \cite{MannaP89}}
\label{fig:temporal_hierarchy}
\end{subfigure}%
\begin{subfigure}[c]{0.51\columnwidth}
  \begin{center}
  \small
	\begin{tikzpicture}[x=1cm,y=0.75cm,outer sep=2pt]

    \node (padding1) at ( 0,0.3) {};
	\node (1) at ( 0, 0) {$\Delta_2$};
	\node (2) at ( 1,-1) {$\Pi_2$};
	\node (3) at (-1,-1) {$\Sigma_2$};
    \node (4) at ( 0,-2) {$\Delta_1$};
	\node (5) at ( 1,-3) {$\Pi_1$};
	\node (6) at (-1,-3) {$\Sigma_1$};
    \node (padding2) at ( 0,-3.2) {};
    
    \path
	(2) edge[draw=none] node[sloped]{$\supset$} (1)
    (3) edge[draw=none] node[sloped]{$\subset$} (1)
    
    (4) edge[draw=none] node[sloped]{$\subset$} (2)
    (4) edge[draw=none] node[sloped]{$\supset$} (3)
    
    (5) edge[draw=none] node[sloped]{$\supset$} (4)
    (6) edge[draw=none] node[sloped]{$\subset$} (4);

%
%

	\end{tikzpicture}
  \end{center}
\subcaption{Syntactic-future hierarchy}
\label{fig:syntactic-future}
\end{subfigure}
\caption{Both hierarchies, side-by-side, indicating the correspondence of \Cref{thm:hierarchy:correspondence}}
\label{fig:hierarchies}
\end{figure}
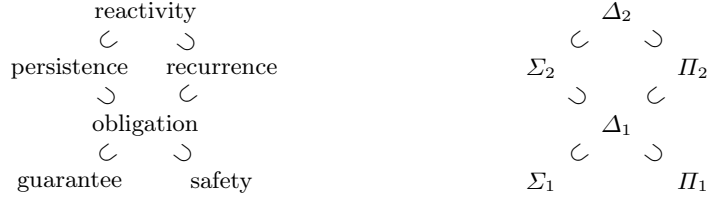

Together with the result of \cite{LPZ85}, stating that every formula of LTL is equivalent to a reactivity formula, Chang, Manna, and Pnueli obtain:

\begin{theorem}[Normalization Theorem \cite{LPZ85,MannaP89,ChangMP92}]
Every LTL formula is equivalent to a formula of $\Delta_2$.
\end{theorem}

In \cite{SickertE20}, Sickert and Esparza obtain a new proof of the Normalization Theorem. They show that every formula $\varphi$ is equivalent to a formula of the form
$$\bigvee_{M \subseteq \sfmu(\varphi), N \subseteq \sfnu(\varphi)} \varphi_{M,N}$$
\noindent where $\sfmu(\varphi)$ and $\sfnu(\varphi)$ are the sets of subformulas of $\varphi$ with top operator in $\{\U, \M\}$ and $\{ \W, \R\}$, respectively, and $\varphi_{M,N}$ is a $\Delta_2$-formula obtained from $\varphi$, $M$, and $N$ by means of a few syntactic rewrite rules. Further, $\varphi_{M,N}$ is at most exponentially longer than $\varphi$. While this is a big improvement with respect to previous procedures, it requires to iterate over all subsets of $\sfmu(\varphi)$ and $\sfnu(\varphi)$, and so the procedure \emph{always} takes exponential time, even for simple families of formulas that have equivalent $\Delta_2$-formulas with only a linear blow-up.

\begin{example} \label{example:exp}
Consider the family of formulas 
$$\varphi_n =  (\cdots ((((a_0 \U a_1) \W a_2) \U a_3) \U a_4) \cdots \U a_n)$$
\noindent for $n \geq 3$. The sets $\sfmu(\varphi_n)$ and $\sfnu(\varphi_n)$ have size $n-1$ and 1, respectively. The procedure of \cite{SickertE20} yields a disjunction of $2^{n+1}$ formulas $\varphi_{M,N}$, and so it takes exponential time in $n$. However,  exhaustive application of a few simplification rules yields a short formula in normal form of length $\Theta(n)$:
$$\begin{array}{rcl}
\varphi_n & \equiv & (\G\F a_1 \wedge   (\cdots ((((a_0 \U a_1) \U ( a_2 \vee \G(a_0 \vee a_1))) \U a_3) \U a_4) \cdots \U a_n) \\[0.1cm]
& & \vee \; (\cdots ((((a_0 \U a_1) \U a_2) \U a_3) \U a_4) \cdots \U a_n)
\end{array}$$
Intuitively, in order to normalize $\varphi_n$ it suffices to solve the ``local'' problem caused by  the subformula $((a_0 \U a_1) \W a_2) \U a_3$ of $\varphi_n$, which is in $\Sigma_3$; however, the procedure of \cite{SickertE20} is blind to this  fact, and generates $2^{n+1}$ formulas, only to simplify them away later on.
\end{example}

\section{A Normalizing Rewrite System} \label{sec:main}

We present a rewrite system that allows us to normalize every LTL formula. As a corollary, we obtain an alternative proof of the Normalization Theorem. 

The key idea is to treat the combinations $\G\F$ (infinitely often) and $\F\G$ (almost always) of temporal operators as \emph{atomic} operators $\GF$ and $\FG$ (notice the typesetting with the two letters touching each other). We call them the \emph{limit operators}; intuitively, whether a word satisfies a formula $\GF \varphi$ or $\FG \varphi$ depends only on its behaviour ``in the limit'', in the sense that $w'w$ satisfies $\GF \varphi$ or $\FG \varphi$ if{}f $w$ does.\footnote{Limit operators are called \emph{suspendable} in \cite{babiak13}.} 
So we add the limit operators to the syntax.
Moreover, in order to simplify the presentation, we also temporarily remove the operators $\M$ and $\R$ (we reintroduce them in Section \ref{sec:extensions}). So we define:

\begin{definition}
Extended LTL formulas over a set $Ap$ of atomic propositions are generated by the syntax:
\begin{align*}
\varphi \Coloneqq \; & \true \mid \false \mid a \mid \neg a \mid \varphi \wedge \varphi \mid \varphi\vee\varphi \\ 
                     & \mid \X\varphi \mid \varphi\U\varphi \mid \varphi\W\varphi \mid \GF \varphi \mid \FG \varphi
\end{align*}
\end{definition}

\noindent When determining the class of a formula in the syntactic future hierarchy, $\GF$ and $\FG$  are implicitly replaced by $\G\F$ and $\F\hspace{0.02cm}\G$. For example, $\F \GF a$ is rewritten into $\F\G\F a$, and so it is a formula of $\Sigma_3$.
In the rest of the section we only consider extended formulas which are by construction negation normal form and call them just formulas. 

Let us now define the precise shape of our normal form, which is a bit more strict than $\Delta_2$.  Formulas of the form $\varphi \U \psi$, $\varphi \W \psi$, $\X\varphi$, $\GF\varphi$, and $\FG\varphi$ are called  \U-, \W-, \X-, \GF-, and \FG-formulas, respectively.
We refer to these formulas as \emph{temporal} formulas. 
The syntax tree $T_\varphi$ of a formula $\varphi$ is defined in the usual way, and $\nnodes{\varphi}$ denotes the number of nodes of $T_\varphi$.
A node of $T_\varphi$ is a \emph{\U-node} if the subformula rooted at it is a \U-formula. \W-, \GF-, \FG- and temporal nodes are defined analogously. 

\newcommand\normalFormConds{\begin{enumerate}
\item No \U-node is under a \W-node.
\item No limit node is under another temporal node.
\item No \W-node is under a \GF-node, and no \U-node is under a \FG-node.
\end{enumerate}}

\begin{definition}
\label{def:normalform}
Let $\varphi$ be an LTL formula. 
A node of $T_\varphi$ is a \emph{limit node} if it is either a \GF-node or a \FG-node. The formula
$\varphi$ is in \emph{normal form} if $T_\varphi$ satisfies the following properties:
\normalFormConds
\end{definition}

\begin{restatable}{remark}{normalFormCharacterization}
Observe that formulas in normal form belong to $\Delta_2$. Even a slightly stronger statement holds: a formula in normal form is a positive Boolean combination of formulas of $\Sigma_2$ and formulas of the form $\GF \psi$ such that $\psi \in \Sigma_1$ (and so $\GF \psi \in \Pi_2$). 

\smallskip

There is a dual normal form in which property 1. is replaced by ``no \W-node is under a \U-node'', and the other two properties do not change. Formulas in dual normal form are positive Boolean combination of formulas of $\Pi_2$ and formulas of the form $\FG \psi$ such that $\psi \in \Pi_1$. Once the Normalization Theorem for the primal normal form is proved, a corresponding theorem for the dual form follows as an easy corollary (see Section \ref{sec:extensions}).
\end{restatable}

In the following three subsections we incrementally normalize formulas by dealing with the three requirements of the normal form one by one. Intermediate normal forms are obtained between stages, which we define formally using the following two measures:

\begin{itemize}
\item $\ubw{\varphi}$ is the number of \U-nodes in $T_\varphi$ that are under some \W-node, but not under any limit node of $T_\varphi$. For example,  if $\varphi = (a \U b) \W (\FG (c \U d))$ then $\ubw{\varphi}=1$.
\item $\gfba{\varphi}$ is the number of distinct limit subformulas under some temporal operator. Formally,
$\gfba{\varphi}$ is the number of limit formulas $\psi'$ such that $\psi'$ is a proper subformula of a temporal subformula (proper or not) of $\varphi$. For example, if $\varphi = (\FG a \, \U \, \GF b) \vee (\GF b \, \W \, \FG a)$ then $\gfba{\varphi}=2$.

\end{itemize}

\begin{definition}
\label{def:1stnormalform}
An LTL formula $\varphi$ is in \emph{\fnf} if $\ubw{\varphi} = 0$, and in \emph{\snf} if $\ubw{\varphi} = 0$ and $\gfba{\varphi} = 0$.
\end{definition}

\noindent We proceed in three stages:

\begin{enumerate}
	\item We remove all \U-nodes that are under some \W-node, but not under any  limit node. The resulting formula is in \fnf.
	\item We remove all limit nodes under some other temporal node. The resulting formula is in \snf.
	\item We remove all \W-nodes under some \GF-node, and all \U-nodes under some \FG-node. The resulting formula is in normal form (\Cref{def:normalform}).
\end{enumerate}

\subsection*{Stage 1: Removing \U-nodes under \W-nodes.}

We consider formulas $\varphi$ with placeholders, i.e., ``holes'' that can be filled with a formula. 
Formally, let $\hole$ be a symbol denoting a special atomic proposition. 
A formula with placeholders is a formula with one or more occurrences of $\hole$, all of them positive (i.e., the formula has  no occurrence of $\neg\hole$. We denote by $\varphi[\psi]$ the result of filling each placeholder of $\varphi$ with an occurrence of $\psi$; formally, $\varphi[\psi]$ is the result of substituting $\psi$ for $\hole$ 
in $\varphi$. For example, if $\varphi\hole = (\hole \W (a \U \hole))$, then $\varphi[\X b] = (\X b) \W (a \U \X b)$. We assume that $\hole$ binds more strongly than any operator, e.g. $\varphi_1 \W \varphi_2[\psi] = \varphi_1 \W (\varphi_2[\psi])$.

\smallskip
This lemma, proved in the Appendix,  allows us to pull \U-subformulas out of \W-formulas:

\begin{restatable}{lemma}{lemmaUWandWU}
\label{lem:UWandWU}
\begin{eqnarray}
\varphi_1 \W \varphi_2[\psi_1 \U \psi_2] & \equiv & (\varphi_1 \U \varphi_2[\psi_1 \U \psi_2]) \vee \G \varphi_1 \label{eqWU} \\[1ex]
\varphi_1[\psi_1 \U \psi_2] \W \varphi_2 & \equiv & (\GF \! \psi_2 \wedge \varphi_1[\psi_1 \W \psi_2] \W \varphi_2) \label{eqUW} \\ & & ~ \vee ~ \varphi_1[\psi_1 \U \psi_2] \U (\varphi_2 \vee (\G \varphi_1[\false]))  \nonumber 
\end{eqnarray}
\end{restatable}

\begin{proposition} \label{prop:23norm}
For every LTL formula $\varphi$ there exists an equivalent formula $\varphi'$ in \fnf{} such that $\nnodes{\varphi'} \leq 4^{2\nnodes{\varphi}} \cdot \nnodes{\varphi}$. Moreover, for every subformula $\GF \psi$ of $\varphi'$ the formula $\psi$ is a subformula of $\varphi$, and every \FG-subformula of $\varphi'$ is also a subformula of $\varphi$.
\end{proposition}

\begin{proof}

We associate to each formula a rank, defined by $\rank{\varphi} = \nnodes{\varphi} + \ubw{\varphi}$. Observe that a formula $\varphi$ is in \fnf{} if{}f $\rank{\varphi} = \nnodes{\varphi}$.
Throughout the proof we say that a formula $\varphi'$ \emph{satisfies the limit property} if for every subformula $\GF \psi$ of $\varphi'$ the formula $\psi$ is a subformula of $\varphi$ and every \FG-subformula of $\varphi'$ is also a subformula of $\varphi$ (notice the asymmetry). Further, we say that a formula $\varphi'$ \emph{satisfies the size property} if $\nnodes{\varphi'} \leq 4^{\mathit{rank}(\varphi)} \cdot \nnodes{\varphi}$ from which the claimed size bound immediately follows.

We prove by induction on $\rank{\varphi}$ that $\varphi$ is equivalent to a formula $\varphi'$ in \fnf{} satisfying the limit and size properties. Within the inductive step we proceed by a case distinction of $\varphi$:

\medskip\noindent If $\varphi = \true, \false, \GF \psi, \FG \psi$ then $\varphi$ is already in \fnf{}, and satisfies the limit and size properties.

\medskip\noindent If $\varphi= \varphi_1 \wedge \varphi_2, \varphi_1 \vee \varphi_2, \varphi_1 \U \varphi_2$ then by induction hypothesis $\varphi_1$ and $\varphi_2$  can be normalized into formulas $\varphi_1'$ and $\varphi_2'$ satisfying the limit and size properties. The formulas $\varphi_1'  \wedge \varphi_2'$, $\varphi_1' \vee \varphi_2'$,  $\varphi_1'  \U \varphi_2'$ are then in \fnf{} (the latter because the additional \U-node is above any \W-node) and satisfy the limit property. The size property holds because:
	
\[\begin{array}{rl}
\nnodes{\varphi_1'} + \nnodes{\varphi_2'} + 1 & \leq 4^{\rank{\varphi_1}} \cdot \nnodes{\varphi_1} + 4^{\rank{\varphi_2}} \cdot \nnodes{\varphi_2} + 1 \\
                               				  & \leq 4^{\rank{\varphi_1}+\rank{\varphi_2}} \cdot (\nnodes{\varphi_1} + \nnodes{\varphi_2} + 1) \\
                               				  & \leq 4^{\rank{\varphi}} \cdot \nnodes{\varphi}
\end{array}\]

\medskip\noindent If $\varphi  = \X \varphi_1$, then by induction hypothesis there is a formula $\varphi_1'$ equivalent to $\varphi_1$ in \fnf{}, and so $\varphi$ is equivalent to $\X \varphi_1'$, which is in \fnf{} and satisfies the limit and size properties.

\medskip\noindent If $\varphi = \varphi_1 \W \varphi_2$ and $\ubw{\varphi} = 0$, then $\varphi_1 \W \varphi_2$ is already in \fnf{} and satisfies the limit and size properties.

\medskip\noindent If $\varphi = \varphi_1 \W \varphi_2$ and $\ubw{\varphi} > 0$, then we proceed by a case distinction:

\begin{itemize}
\item $\varphi_2$ contains at least one \U-node that is not under a limit node. Let $\psi_1 \U \psi_2$ be such a \U-node. We derive $\varphi_2\hole$ from $\varphi_2$ by replacing each \U-node labelled by $\psi_1 \U \psi_2$ by the special atomic proposition $\hole$. By \cref{lem:UWandWU}(\ref{eqWU}) we have:
	\[ \varphi_1 \W \varphi_2[\psi_1 \U \psi_2] \equiv \varphi_1 \U \varphi_2[\psi_1 \U \psi_2] \vee \varphi_1 \W \false\]
Since $\rank{\varphi_1} < \rank{\varphi}$, $\rank{\varphi_2} <  \rank{\varphi}$, and $\rank{\varphi_1 \W \false} < \rank{\varphi}$ (the latter because $\varphi_2$ contains at least one \U-node), by induction hypothesis $\varphi_1$,  $\varphi_2$, and $\varphi_1 \W \false$ can be normalized into formulas $\varphi_1'$, $\varphi_2'$, and $\varphi_3'$ satisfying the limit and size properties. So $\varphi$ can be normalized into $\varphi' = \varphi_1' \U \varphi_2' \vee \varphi_3'$. Moreover, $\varphi'$ satisfies the limit property, because all \GF- and \FG-subformulas of $\varphi'$ are subformulas of $\varphi_1'$, $\varphi_2'$, or $\varphi_3'$. For the size property we calculate:

\[\begin{array}{rl}
\nnodes{\varphi'} & = \nnodes{\varphi_1'} + \nnodes{\varphi_2'} + \nnodes{\varphi_3'} + 2 \\ 
                  & \leq 4^{\rank{\varphi_1}} \cdot \nnodes{\varphi_1} + 4^{\rank{\varphi_2}} \cdot \nnodes{\varphi_2} + 4^{\rank{\varphi_1\W \false}} \cdot \nnodes{\varphi_1\W\false} + 2\\
                  & \leq 4^{\rank{\varphi} - 1} \cdot (\nnodes{\varphi_1} + \nnodes{\varphi_2} + \nnodes{\varphi_1\W\false} + 2) \\
                  & \leq 4^{\rank{\varphi} - 1} \cdot 4 \cdot \nnodes{\varphi} = 4^{\rank{\varphi}} \cdot \nnodes{\varphi}
\end{array}\]

\item Every \U-node of $\varphi_2$ is under a limit node, and $\varphi_1$ contains at least one \U-node that is not under any  limit node. Then $\varphi_1$ contains a maximal subformula $\psi_1 \U \psi_2$ (with respect to the subformula order) that is not under a limit node. We derive $\varphi_1\hole$ from $\varphi_1$ by replacing each \U-node labelled by $\psi_1 \U \psi_2$ that does not appear under a limit node by the special atomic proposition $\hole$. By \cref{lem:UWandWU}(\ref{eqUW}), we have
\begin{multline*}
\varphi_1[\psi_1 \U \psi_2] \W \varphi_2 \equiv \\ \big(\GF \psi_2 \wedge \underbrace{\varphi_1[\psi_1 \W \psi_2] \W \varphi_2}_{\rho_1}\big) \vee \big( \underbrace{\varphi_1[\psi_1 \U \psi_2]}_{\rho_2}  \U  (\underbrace{\varphi_2 \vee (\varphi_1[\false] \W \false}_{\rho_3}) \big)
\end{multline*}
In order to apply the induction hypothesis we argue that $\rho_1$, $\rho_2$, and $\rho_3$ have rank smaller than $\varphi$, and thus can be normalized to $\rho_1'$, $\rho_2'$ and $\rho_3'$ satisfying the limit and size properties. The formula $\rho_1$ has the same number of nodes as $\varphi$, but fewer $\U$-nodes under $\W$-nodes; so $\ubw{\rho_1} < \ubw{\varphi}$ and thus $\rank{\rho_1} < \rank{\varphi}$. The same argument applies to $\rho_3$.
Finally, $\rank{\rho_2} < \rank{\varphi}$ follows from the fact that $\rho_2$ has fewer nodes than $\varphi$. So $\varphi$ can be normalized to $\varphi'=(\GF\psi_2 \wedge \rho_1') \vee (\rho_2' \U \rho_3')$.

We show that $\varphi'$ satisfies the limit property. Let $\GF \psi$ be a subformula of $\varphi'$. If  $\GF \psi = \GF \psi_2$, then we are done, because $\psi_2$ is a subformula of $\varphi$. Otherwise $\GF\psi$ is a subformula of $\rho_1'$, $\rho_2'$, or $\rho_3'$. Since all of them satisfy the limit property, $\psi$ is a subformula of $\varphi$, and we are done. Further, every $\FG$-subformula of $\varphi'$ belongs to $\rho_1'$, $\rho_2'$, or $\rho_3'$ and so it is also subformula of $\varphi$. For the size property we calculate:
\begin{align*}
\nnodes{\varphi'} & = \nnodes{\rho_1'} + \nnodes{\rho_2'} + \nnodes{\rho_3'} + \nnodes{\psi_2} + 4 \\ 
                  & \leq 4^{\rank{\rho_1}} \cdot \nnodes{\rho_1} + 4^{\rank{\rho_2}} \cdot \nnodes{\rho_2} + 4^{\rank{\rho_3}} \cdot \nnodes{\rho_3} + \nnodes{\varphi_1} + 4\\
                  & \leq 4^{\rank{\varphi} - 1} \cdot (\nnodes{\varphi} + \nnodes{\varphi} + \nnodes{\varphi})  + \nnodes{\varphi} + 4  \\
                  & \leq 4^{\rank{\varphi} - 1} \cdot 4 \cdot \nnodes{\varphi} = 4^{\rank{\varphi}} \cdot \nnodes{\varphi}
\end{align*}
\end{itemize}
%
%
\end{proof}

\subsection*{Stage 2: Moving $\GF$- and $\FG$-subformulas up.}

In this section, we address the second property of the normal form. The following lemma allows us to pull limit subformulas out of any temporal formula. (Note that the second rule is only necessary if the formula before stage 1 contained \FG-subformulas, since stage 1 only creates new \GF-formulas.)

\begin{restatable}{lemma}{lemmaGFI}
\label{lemma:GF1}
\begin{eqnarray} 
\varphi[\GF \psi] & \equiv & (\GF \psi \wedge \varphi[\true] ) \vee \varphi[\false]  \label{eqGF1} 
\\ \varphi[\FG \psi] & \equiv & (\FG \psi \wedge \varphi[\true] ) \vee \varphi[\false]  \label{eqFG1}
\end{eqnarray} 
\end{restatable}

We show using (\ref{eqGF1}) and (\ref{eqFG1}) that every formula in \fnf{} can be transformed into an equivalent formula in \snf.



\begin{proposition} \label{prop:snf}
Every LTL formula $\varphi$ in \fnf{} is equivalent to a formula $\varphi'$ in \snf{} such that $\nnodes{\varphi'} \leq 3^{\gfba{\varphi}} \cdot \nnodes{\varphi}$. Moreover, the size of the limit subformulas does not increase: for every $b > 0$, if $\nnodes{\psi} \leq b$ for every limit subformula of $\varphi$, then $\nnodes{\psi'} \leq b$ for every limit subformula of $\varphi'$.
\end{proposition}

\begin{proof}
We proceed by induction on the number of proper limit sub\-formulas of $\varphi$. If $\varphi$ does not contain any, then it is already in \snf. Assume there exists such a proper limit subformula $\psi$ that is smaller (or incomparable) to all other limit subformulas of $\varphi$ according to the subformula order. We derive $\varphi\hole$ from $\varphi$ by replacing each limit-node labelled by $\psi$ by the special atomic proposition $\hole$. We then apply \cref{lemma:GF1} to obtain:
\[ \varphi[\psi] \equiv (\psi \wedge \varphi[\true] ) \vee \varphi[\false] \qquad \mbox{ where } \psi = \GF \psi', \FG \psi'  \ .\]
Note that $\psi$ does not properly contain any  limit subformula, and so it is in \snf{}. Both $\varphi[\true]$ and $\varphi[\false]$ are still in \fnf{} and they have one limit operator less than $\varphi$. Thus they can be normalized by the induction hypothesis into $\varphi_1'$ and $\varphi_2'$ in \snf. Finally, $\varphi' = (\psi \wedge \varphi_1') \vee \varphi_2'$ is a Boolean combination of formulas in \snf, so it is in \snf.	The number of nodes of $T_{\varphi'}$ can be crudely bounded as follows:
\begin{equation*}
\begin{array}{rl}
	\nnodes{\varphi'} & \leq \nnodes{\varphi_1'} + \nnodes{\varphi_2'} + \nnodes{\psi} + 2 \\[1ex]
	                  & \leq 2 \cdot 3^{\gfba{\varphi[\true]}} \cdot \nnodes{\varphi[\true]} + \nnodes{\psi} + 2 \\[1ex]
	                  & \leq 3^{\gfba{\varphi[\true]}} \cdot \big(2 \cdot (\nnodes{\varphi} - \nnodes{\psi} + 1) + \nnodes{\psi} + 2\big)\\[1ex]
	                  & \leq 3^{\gfba{\varphi} - 1} \cdot (2	 \nnodes{\varphi} - \nnodes{\psi} + 4)\\[1ex]
	                  & \leq 3^{\gfba{\varphi} - 1} \cdot (3 \nnodes{\varphi}) = 3^{\gfba{\varphi}} \cdot \nnodes{\varphi}
\end{array}
\end{equation*}
where the induction hypothesis is used in the second inequality, and $\nnodes{\psi} \geq 2$ and $\nnodes{\varphi} \geq 2$ in the last one.

To show that the size of the limit subformulas does not increase, let $b$ be a bound on the size of the \GF-subformulas of $\varphi$. We claim that the size of each \GF-subformula of $\varphi'$ is also bounded by $b$ (the case of $\FG \psi$ is analogous). Indeed, the \GF-subformulas of $\varphi'$ are $\psi$ (which is already in $\varphi$) and the \GF-subformulas of $\varphi_1'$ and $\varphi_2'$. Since the \GF-subformulas of $\varphi[\true]$ and $\varphi[\false]$ can only have decreased in size, by induction hypothesis the number of nodes of any \GF-subformula of $\varphi_1'$ and $\varphi_2'$ is bounded by $b$, and we are done.
\end{proof}


\subsection*{Stage 3: Removing $\W$-nodes ($\U$-nodes)  under $\GF$-nodes ($\FG$-nodes)}

The normalization of LTL formulas is completed in this section by fixing the problems within limit subformulas. In order to do so, we introduce two new rewrite rules that allow us to pull \W-subformulas out of \GF-formulas, and \U-subformulas out of \FG-formulas.

\begin{restatable}{lemma}{lemmaGFII}
\label{lemma:GF2}
\begin{eqnarray} 
\GF \varphi [\psi_1 \W \psi_2]  & \equiv & \GF \varphi [\psi_1 \U \psi_2] \vee  (\FG \psi_1 \wedge \GF \varphi [\true])  \;\;\; \label{eqGF2} \\
\FG \varphi[\psi_1 \U \psi_2] & \equiv & (\GF \psi_2 \wedge \FG \varphi[\psi_1 \W \psi_2]) \vee \FG \varphi[\false]  \;\;\; \label{eqFG2}
\end{eqnarray}
\end{restatable}

The following proposition repeatedly applies these rules to show that limit formulas can be normalized with an exponential blowup.

\begin{proposition} \label{prop:gffg}
For every LTL formula $\varphi$ without limit operators, $\GF \varphi$ and $\FG \varphi$ can be normalized into formulas with at most $\nnodes{\varphi'} \leq 3^{\nnodes{\varphi}} \cdot \nnodes{\varphi}$ nodes.
\end{proposition}

\begin{proof}
A \GF-obstacle of a formula is a \W-node or a \U-node under a \W-node inside a \GF-node. Similarly, a \FG-obstacle is a \U-node or a \W-node under a \U-node inside a \FG-node. Finally, an obstacle is either a \GF-obstacle or an \FG-obstacle. We proceed by induction on the number of obstacles of $\GF \varphi$ or $\FG \varphi$. If they have no obstacles, then they are already in normal form (\cref{def:normalform}).

Assume $\GF \varphi$ has at least one obstacle. Then $\varphi$ contains at least one maximal \W-node $\psi_1 \W \psi_2$. We derive $\GF \varphi\hole$ from $\GF \varphi$ by replacing each \W-node labelled by $\psi_1 \W \psi_2$ by the special atomic proposition $\hole$. By \cref{eqGF2}, $\GF \varphi[\psi_1 \W \psi_2]$ is equivalent to
\[  \GF \varphi [\psi_1 \U \psi_2] \vee  (\FG \psi_1 \wedge \GF \varphi [\true]) \]
We claim that each of $\GF \varphi [\psi_1 \U \psi_2]$, $\GF \varphi [\true]$, and $\FG \psi_1$ has fewer obstacles than 
$\GF \varphi[\psi_1 \W \psi_2]$, and so can be normalized by induction hypothesis.  Indeed, $\GF \varphi [\psi_1 \U \psi_2]$, and $\GF \varphi[\true]$ have at least one \W-node less than $\varphi$, and the number of \U-nodes under a \W-node, due to the maximality of $\psi_1 \W \psi_2$, has not increased, and as a consequence it has fewer \GF-obstacles (and by definition no \FG-obstacles). For $\FG \psi_1$, observe first that every \FG-obstacle of $\FG \psi_1$ is a \GF-obstacle of $\GF \varphi[\psi_1 \W \psi_2]$. Indeed, the obstacles of $\FG \psi_1$ are the \U-nodes and the \W-nodes under \U-nodes; the former were under \W-nodes in $\varphi$, and the latter were \W-nodes of $\varphi$, and so both \GF-obstacles of $\GF\varphi$. Moreover, $\psi_1 \W \psi_2$ is a \GF-obstacle of $\varphi$, but not a \FG-obstacle of $\FG \psi_1$. Hence, the number of obstacles has decreased. 

Assume now that $\FG \varphi$ has at least one obstacle. Then $\varphi$ contains at least one maximal \U-node $\psi_1 \U \psi_2$. We derive $\FG \varphi\hole$ from $\FG \varphi$ by replacing each \U-node labelled by $\psi_1 \U \psi_2$ by the special atomic proposition $\hole$. By \cref{eqFG2}, $\FG \varphi[\psi_1 \U \psi_2]$ is equivalent to
\[ (\GF \psi_2 \wedge \FG \varphi[\psi_1 \W \psi_2]) \vee \FG \varphi[\false] \]
Each of $\GF \psi_2$, $\FG \varphi[\psi_1 \W \psi_2]$, and $\FG \varphi[\false]$ has fewer obstacles as $\FG \varphi[\psi_1 \U \psi_2]$, and can be normalized by induction hypothesis. The proof is as above.

	The size of the formula increases at most by a factor of 3 on each step, and the number of steps is bounded by the number of both \W-nodes and \U-nodes in $\varphi$, which is bounded by the total number of nodes in $\varphi$. So the formula has at most $3^{\nnodes{\varphi}} \nnodes{\varphi}$ nodes.
\end{proof}

\subsection*{The Normalization Theorem}
\label{sec:overview}

	The main result directly follows from the previous propositions.

\begin{theorem}\label{thm:main}
Every formula $\varphi$ of LTL is normalizable into a formula with at most $4^{7\nnodes{\varphi}}$ nodes.
\end{theorem}

\begin{proof}
Any LTL formula $\varphi$ can be transformed into an equivalent $\varphi'$ in \fnf{} of size $\nnodes{\varphi'} \leq 4^{2\nnodes{\varphi}} \cdot \nnodes{\varphi}$ by \cref{prop:23norm}. Moreover, $\gfba{\varphi'} \leq 2 \cdot \nnodes{\varphi}$, since every \FG-subformula of $\psi'$ and every argument $\psi$ of a \GF-subformula of $\varphi'$ is a subformula of $\varphi$. In addition, $\nnodes{\psi} \leq \nnodes{\varphi}$ for every $\GF \psi$ subformula of $\varphi$.

	According to \cref{prop:snf}, for every formula $\varphi'$ in \fnf{} there is an equivalent formula $\varphi''$ in \snf{} with
\begin{equation}
	\label{eq:phi2nd}
	\nnodes{\varphi''} \leq 3^{\gfba{\varphi'}} \cdot \nnodes{\varphi'} \leq 3^{2 \nnodes{\varphi}} \cdot (4^{2\nnodes{\varphi}} \cdot \nnodes{\varphi}) \leq  3^{2\nnodes{\varphi}} \cdot 4^{3\nnodes{\varphi}}
	\tag{$\star$}
\end{equation}
This formula is a Boolean combination of  limit formulas with at most $\nnodes{\varphi}$ nodes, not containing any proper limit node, and other temporal formulas containing neither limit nodes nor \U-nodes under \W-nodes. The latter are in $\Sigma_2$ and \cref{prop:gffg} deals with the former. Notice that every $\GF \psi$ and $\FG \psi$ subformula has at most $\nnodes{\varphi}$ nodes and thus  can be normalized into a formula with at most $3^{\nnodes{\varphi}} \nnodes{\varphi}$ nodes. The result $\varphi'''$ of replacing these limit subformulas by their normal forms within $\varphi''$ is a Boolean combination of normal forms, and so we are done. The number of nodes in the resulting formula $\varphi'''$ is at most:
\begin{align*}
\nnodes{\varphi'''} \leq & ~ \nnodes{\varphi''} + \gfba{\varphi''} \cdot 3^{\nnodes{\varphi}} \cdot \nnodes{\varphi}
                    & \gfba{\varphi''} \leq |\varphi''| \\
		    \leq & ~ \nnodes{\varphi''} \cdot 3^{2\nnodes{\varphi}} \cdot (\nnodes{\varphi} + 1)
		    & (\text{\ref{eq:phi2nd}}) \\
		    \leq & ~ 4^{3\nnodes{\varphi}} \cdot 3^{4\nnodes{\varphi}} \cdot (\nnodes{\varphi} + 1)
		    & \nnodes{\varphi} + 1 \leq 4^{\nnodes{\varphi}/2} \\
		    \leq & ~ 4^{3\nnodes{\varphi}} \cdot 4^{(4 \log_4 3 + \frac12) \nnodes{\varphi}} \leq 4^{7\nnodes{\varphi}}
\end{align*}   
\end{proof}

\section{Summary of the normalization algorithm} \label{sec:summary}

We summarize the steps of the normalization algorithm described and proven in~\cref{sec:main}. Recall that a formula is in normal form if{}f it satisfies the following properties:
\normalFormConds
The normalization algorithm applies the rules in~\cref{tab:allrules} as follows to fix any violation of these properties:
\begin{enumerate}
	\item \U-nodes under \W-nodes and not under limit nodes are removed using rules (\ref{eqWU}) and (\ref{eqUW}). This may introduce new \GF-subformulas. By applying (\ref{eqUW}) only to highest \U-nodes of $\varphi_1$ the number of new \GF-subformulas is only linear in the size of the original formula.
	\item Limit nodes under other temporal nodes are pulled out using rules (\ref{eqGF1}) and (\ref{eqFG1}). By applying the rules only to the lowest limit nodes, it only needs to be applied once for each limit subformula. 
	\item \W-nodes under \GF-nodes are removed using rule (\ref{eqGF2}), and \U-nodes under \FG-nodes are removed using rule (\ref{eqFG2}). This may produce new limit nodes of smaller size that are handled recursively. Choosing highest \W- and \U-nodes ensures that the process produces only a single exponential blowup over the initial size of the formula.
\end{enumerate}
{
\begin{table}[tb]
\begin{equation*}
\begin{array}{lcrcl@{\;\;\;}}
\multirow{3}*{\text{Stage 1:}~} & (\ref{eqWU}) \quad & \varphi_1 \W \varphi_2[\psi_1 \U \psi_2] & \equiv & \varphi_1 \U \varphi_2[\psi_1 \U \psi_2] \vee \G \varphi_1  \\[1ex]
	                  & (\ref{eqUW}) \quad & \varphi_1[\psi_1 \U \psi_2] \W \varphi_2 & \equiv & (\GF \psi_2 \wedge \varphi_1[\psi_1 \W \psi_2] \W \varphi_2)   \\ & & & & ~ \vee ~ \varphi_1[\psi_1 \U \psi_2] \U (\varphi_2 \vee \G \varphi_1[\false])  \\[1em] \hline \\
\multirow{2}*{\text{Stage 2:}~} &  (\ref{eqGF1}) \quad & \varphi[\GF \psi] & \equiv & (\GF \psi \wedge \varphi[\true]) \vee \varphi[\false]  \\[1ex]
	&  (\ref{eqFG1}) \quad & \varphi[\FG \psi] & \equiv & (\FG \psi \wedge \varphi[\true]) \vee \varphi[\false]  \\[1em] \hline \\
\multirow{2}*{\text{Stage 3:}~} & (\ref{eqGF2}) \quad & \GF \varphi [\psi_1 \W \psi_2] & \equiv & \GF \varphi [\psi_1 \U \psi_2] \vee  (\FG \psi_1 \wedge \GF \varphi [\true]) \\[1ex]
	& (\ref{eqFG2})  \quad & \FG \varphi[\psi_1 \U \psi_2] & \equiv & (\GF \psi_2 \wedge \FG \varphi[\psi_1 \W \psi_2]) \vee \FG \varphi[\false]  \\[1ex]
\end{array}
\end{equation*}
\caption{Normalization rules.} \label{tab:allrules}
\end{table}
}

After the three steps, a formula in normal form is obtained with a single exponential blowup in the number of nodes.

Moreover, notice that $\psi_1$ itself does not play any role in rules (2) and (6), and neither does $\psi_2$ in (5). Hence, the application of (1) can be made mode efficient by replacing not only every occurrence of $\psi_1 \U \psi_2$ outside a limit subformula with $\psi_1 \W \psi_2$ and $\false$, but also every occurrence of $\psi \U \psi_2$ for any formula $\psi$ by $\psi \W \psi_2$ and by $\false$. The same holds for rules (5) and (6).

\begin{example}	
Let us apply the procedure to the formula $\varphi_n$ in~\cref{example:exp}. In stage 1, rule (\ref{eqUW}) matches the subformula $(a_0 \U a_1) \W a_2$ 
and rewrites it to $\GF a_1 \wedge (a_0 \W a_1) \W a_2 \vee (a_0 \U a_1) \U (a_2 \vee \false \W \false)$, where $\false \W \false$ can be simplified to $\false$ and removed. The rewritten formula is in 1-form, because there is no \U-node under a \W-node, so we can continue to stage 2. Now, we must pull the \GF-node $\GF a_1$ out the cascade of \U-nodes using rule (\ref{eqGF1}). This yields
$$\begin{array}{rcl}
\varphi_n & \equiv & (\GF a_1 \wedge (\cdots ((((a_0 \W a_1) \W a_2 \vee (a_0 \U a_1) \U a_2) \U a_3) \U a_4) \cdots \U a_n) \\[0.1cm]
& & \vee \; ((\cdots ((((a_0 \U a_1) \U a_2) \U a_3) \U a_4) \cdots \U a_n)
\end{array}$$
Since the only remaining  limit node is outside any temporal formula, we have obtained a formula in 1-2-form and the procedure arrives to stage 3. Again, the only limit subformula is $\GF a_1$, and $a_1$ does not contain any \W-node, so the formula is completely normalized and we have finished. Observe that $\varphi_n$ has been normalized by exactly two rule applications for all $n \geq 3$, so the algorithm proceeds in linear-time for this family of formulas. The result is not identical, but very similar to the one in~\cref{example:exp}.
\end{example}

\begin{table}[tb]
\[\begin{array}{r@{\;\equiv\;}l@{\;\;\;}c}
	\varphi_1[\psi_1 \M \psi_2] \W \varphi_2 & (\GF \psi_1 \wedge \varphi_1[\psi_1 \R \psi_2] \W \varphi_2) \vee \varphi_1[\psi_1 \M \psi_2] \U (\varphi_2 \vee \G \varphi_1[\false]) &  \\[1ex]
	\varphi_1 \W \varphi_2[\psi_1 \M \psi_2] & \varphi_1 \U \varphi_2[\psi_1 \M \psi_2] \vee \G \varphi_1 & \\[1ex]
	\varphi_1[\psi_1 \U \psi_2] \R \varphi_2 & \varphi_1[\psi_1 \U \psi_2] \M \varphi_2 \vee \G \varphi_2 & \\[1ex]
	\varphi_1[\psi_1 \M \psi_2] \R \varphi_2 & \varphi_1[\psi_1 \M \psi_2] \M \varphi_2 \vee \G \varphi_2 & \\[1ex]
	\varphi_1 \R \varphi_2[\psi_1 \U \psi_2] & (\GF \psi_2 \wedge \varphi_1 \R \varphi_2[\psi_1 \W \psi_2]) \vee (\varphi_1 \vee \G \varphi_2[\false]) \M \varphi_2[\psi_1 \U \psi_2] & \\[1ex]
	\varphi_1 \R \varphi_2[\psi_1 \M \psi_2] & (\GF \psi_1 \wedge \varphi_1 \R \varphi_2[\psi_1 \R \psi_2]) \vee (\varphi_1 \vee \G \varphi_2[\false]) \M \varphi_2[\psi_1 \M \psi_2] &  \\[1.5ex]
	\GF \varphi [\psi_1 \R \psi_2] & \GF \varphi [\psi_1 \M \psi_2] \vee  (\FG \psi_2 \wedge \GF \varphi [\true]) & \\[1ex]
	\FG \varphi[\psi_1 \M \psi_2] & (\GF \psi_1 \wedge \FG \varphi[\psi_1 \R \psi_2]) \vee \FG \varphi[\false] &  \\[1ex]
\end{array}\]
\caption{Normalization rules for $\R$ and $\M$.} \label{tab:allrulesRM}
\end{table}

\section{Extensions}
\label{sec:extensions}

\paragraph{The operators $\R$ and $\M$.} We have omitted these operators from the proof and the normalization procedure, since they can be expressed in terms of the subset of operators we have considered. However, this translation exponentially increases the number of nodes of the formula, so handling them directly is convenient for efficiency. Their role at every step of the procedure is analogous to that of the \U\ and \W\ operators, i.e. we treat \R\ in the same as \W\ and we treat \M\ in the same way as $\U$. The corresponding rules are shown in~\cref{tab:allrulesRM}.

\paragraph{Dual normal form.} Recall that a formula is in dual normal form if it satisfies conditions 2. and 3. of Definition \ref{def:normalform} and no \W-node is under a \U-node. Given a formula $\varphi$, let $\overline{\varphi}$ be a formula in negation normal form equivalent to $\neg\varphi$,  and let $\psi$ be a formula in primal normal form equivalent to $\overline{\varphi}$. Since $\varphi \equiv  \neg \overline{\varphi} \equiv \neg \psi$, pushing the negation into $\psi$ yields a formula equivalent to $\varphi$ in dual normal form.

\paragraph{Past LTL.}  Past LTL is an extension of LTL with past operators like yesterday ($\mathbf{Y}$), since ($\mathbf{S}$), etc. In an appendix of \cite{gabbay87}, Gabbay introduced eight rewrite rules to pull future operators out of past operators. Combining these rules with ours yields a procedure that transforms a Past LTL formula into a normalized LTL formula, where past operators are gathered in past-only subformulas, and so can be considered  atomic propositions. 



\section{Experimental Evaluation}
\label{sec:experiments}

We have implemented the normalization procedure summarized in~\cref{sec:summary} as a C++ program,\footnote{The implementation is available at \url{https://github.com/ningit/ltl2delta2rs}.} and compared its performance and the size of the generated formulas with the implementation of the procedure of~\cite{SickertE20} included in the Owl tool~\cite{KretinskyMS18}.\footnote{We evaluate the tool build from commit \texttt{2fb342a09d3a9d7025b219404c764021d17b7ebd} of \url{https://gitlab.lrz.de/i7/owl/}.} In order to make the comparison as fair as possible, we have implemented the same basic simplification rules that are eagerly applied during the normalization process.

We consider the following test suites: TLSF($a$-$b$) is the repertory of formulas of the 2021 Reactive Synthesis Competition of sizes between $a$ and $b$; random formulas are a set of 1000 randomly generated formulas, $\W\U^*$ is the family of formulas in \cref{example:exp} for $2 \leq n \leq 200$; finally, $(\W\U)^*$ is the family defined by $\varphi_0 = a_0$ and $\varphi_{n+1} = (\varphi_n \U a_{2n - 1}) \W a_{2n}$ for $1 \leq n \leq 5$. Notice that the last family is limited to $n = 5$ because Owl cannot handle $\varphi_6$ due to the size of the powersets involved, while this is not a limitation for the new procedure.


\begin{table}[ht]
\centering
\newcommand*\rd[1]{\multirow{2}*{#1}}
\begin{tabular}{l @{\quad} r @{\quad} r @{\quad} r @{\quad} r @{\quad} l}
	\toprule
	\rd{Test cases}  	& \multicolumn3c{Size blowup}        & Time   & \\ \cmidrule{2-4}
				& \begin{tabular}{c}Mean\\(Tree)\end{tabular}      & \begin{tabular}{c}Worst-case\\(Tree)\end{tabular} & \begin{tabular}{c}Worst-case\\(DAG)\end{tabular}   & (ms)   & \\
	\midrule
	\rd{Random formulas}    & 1.38     & 28.47    & 4.64    & 67     & New \\
	                        & 1.06     & 10.69    & 3.90    & 589    & Owl \\
	\rd{$\W\U^*$}           & 2.12     & 3.57     & 2.20    & 90     & New \\
	                        & 2.12     & 4.00     & 2.60    & 32343  & Owl \\
	\rd{$(\W\U)^*$}         & 193.58   & 744.29   & 17.81   & 11     & New \\
	                        & 27.86    & 73.33    & 10.05   & 54     & Owl \\
	\rd{TLSF(-100)}         & 1.04     & 4.60     & 3.14    & 65     & New \\
	                        & 1.17     & 8.02     & 2.48    & 867    & Owl \\
	\rd{TLSF(100-300)}      & 2.14     & 369.66   & 15.10   & 230    & New \\
	                        & 1.14     & 12.54    & 2.47    & 8636   & Owl \\
	\bottomrule \\
\end{tabular}
\caption{Experimental comparison of our normalization procedure and the one of~\cite{SickertE20}.}
\label{table:experiment}
\end{table}
\Cref{table:experiment} shows the mean and worst-case blowup of the syntax tree of the formulas (i.e., the ratio between the sizes of the formulas before and after normalization), the worst-case blowup of their directed acyclic graphs, and the execution time. Generally, the new procedure is faster but generates larger formulas.  However, the execution time and the size of the formulas can be strongly affected by slight changes in the procedure. For example, selecting an innermost instead of an outermost $\psi_1 \U \psi_2$ when matching $\varphi[\psi_1 \U \psi_2]$ in rule (\ref{eqWU}) yields much bigger formulas in some examples, like TLSF(-100), but produces the opposite effect in others, like $(\W\U)^*$. Applying stage 2 separately to each topmost temporal formula is generally better than applying it to the whole term, but it can sometimes be slightly worse. Characterizing these situations and designing a procedure that adapts to them is a subject for future experiments.

\section{Conclusions}
We have presented a simple rewrite system that transforms any LTL formula into an equivalent formula in $\Delta_2$. 
We think that, together with \cite{SickertE20}, this result demystifies the Normalization Theorem of Chang, Manna, and Pnueli, which heavily relied on automata-theoretic results, and involved a nonelementary blowup. Indeed, the only conceptual difference between our procedure and a rewrite system for bringing Boolean formulas in CNF is the use of rewrite rules with contexts. 

The normalization procedure of  Sickert and Esparza has already found applications to the translation of LTL formulas into deterministic or limit-deterministic $\omega$-automata \cite{SickertE20}. Until now normalization had not been considered, because of the non-elementary blow-up, much higher than the double exponential blow-up of existing constructions. With the new procedure, translations that first normalize the formula, and then apply efficient formula-to-automaton procedures specifically designed for formulas in normal form, have become competitive. Our new algorithm, purely based on rewriting rules, makes this even more attractive. More generally, we think that the design of analysis procedures for formulas in normal form (to check satisfiability, equivalence, or other properties) should be further studied in the coming years.

\bibliographystyle{splncs04}
\bibliography{refs}

\begin{thebibliography}{10}
\providecommand{\url}[1]{\texttt{#1}}
\providecommand{\urlprefix}{URL }
\providecommand{\doi}[1]{https://doi.org/#1}

\bibitem{babiak13}
Babiak, T., Badie, T., Duret{-}Lutz, A., Kret{\'{\i}}nsk{\'{y}}, M., Strejcek,
  J.: Compositional approach to suspension and other improvements to {LTL}
  translation. In: {SPIN} 2013. LNCS, vol.~7976, pp. 81--98. Springer (2013).
  \doi{10.1007/978-3-642-39176-7_6}

\bibitem{DBLP:conf/fossacs/BokerLS22}
Boker, U., Lehtinen, K., Sickert, S.: On the translation of automata to linear
  temporal logic. In: FoSSaCS 2022. {LNCS}, vol. 13242, pp. 140--160. Springer
  (2022). \doi{10.1007/978-3-030-99253-8_8}

\bibitem{CernaP03}
Cern{\'{a}}, I., Pel{\'{a}}nek, R.: Relating hierarchy of temporal properties
  to model checking. In: {MFCS} 2003. {LNCS}, vol.~2747, pp. 318--327. Springer
  (2003). \doi{10.1007/978-3-540-45138-9_26}

\bibitem{ChangMP92}
Chang, E.Y., Manna, Z., Pnueli, A.: Characterization of temporal property
  classes. In: {ICALP} 1992. {LNCS}, vol.~623, pp. 474--486. Springer (1992).
  \doi{10.1007/3-540-55719-9_97}

\bibitem{gabbay87}
Gabbay, D.M.: The declarative past and imperative future: Executable temporal
  logic for interactive systems. In: Temporal Logic in Specification. {LNCS},
  vol.~398, pp. 409--448. Springer (1987). \doi{10.1007/3-540-51803-7_36}

\bibitem{KretinskyMS18}
Kret{\'{\i}}nsk{\'{y}}, J., Meggendorfer, T., Sickert, S.: Owl: {A} library for
  {\(\omega\)}-words, automata, and {LTL}. In: {ATVA} 2018. {LNCS}, vol. 11138,
  pp. 543--550. Springer (2018). \doi{10.1007/978-3-030-01090-4_34}

\bibitem{LPZ85}
Lichtenstein, O., Pnueli, A., Zuck, L.D.: The glory of the past. In: Logic of
  Programs. {LNCS}, vol.~193, pp. 196--218. Springer (1985).
  \doi{10.1007/3-540-15648-8_16}

\bibitem{Mal10}
Maler, O.: On the {Krohn}-{Rhodes} cascaded decomposition theorem. In: Time for
  Verification, Essays in Memory of Amir Pnueli. {LNCS}, vol.~6200, pp.
  260--278. Springer (2010). \doi{10.1007/978-3-642-13754-9_12}

\bibitem{MP90}
Maler, O., Pnueli, A.: Tight bounds on the complexity of cascaded decomposition
  of automata. In: Proc.\ of {FOCS}. pp. 672--682 (1990).
  \doi{10.1109/FSCS.1990.89589}

\bibitem{MP94}
Maler, O., Pnueli, A.: On the cascaded decomposition of automata, its
  complexity and its application to logic. Unpublished (1994),
  \url{http://www-verimag.imag.fr/~maler/Papers/decomp.pdf}

\bibitem{MannaP89}
Manna, Z., Pnueli, A.: A hierarchy of temporal properties. In: {PODC}. pp.
  377--410. {ACM} (1990). \doi{10.1145/93385.93442}

\bibitem{MannaPnueli91}
Manna, Z., Pnueli, A.: Completing the temporal picture. Theor. Comput. Sci.
  \textbf{83}(1),  91--130 (1991). \doi{10.1016/0304-3975(91)90041-Y}

\bibitem{PelanekS05}
Pel{\'{a}}nek, R., Strejcek, J.: Deeper connections between {LTL} and
  alternating automata. In: {CIAA} 2005. {LNCS}, vol.~3845, pp. 238--249.
  Springer (2005). \doi{10.1007/11605157_20}

\bibitem{PitermanP18}
Piterman, N., Pnueli, A.: Temporal logic and fair discrete systems. In:
  Handbook of Model Checking, pp. 27--73. Springer (2018).
  \doi{https://doi.org/10.1007/978-3-319-10575-8_2}

\bibitem{Pnueli77}
Pnueli, A.: The temporal logic of programs. In: {FOCS}. pp. 46--57. {IEEE}
  Computer Society (1977). \doi{10.1109/SFCS.1977.32}

\bibitem{Pnueli81}
Pnueli, A.: The temporal semantics of concurrent programs. Theor. Comput. Sci.
  \textbf{13},  45--60 (1981). \doi{10.1016/0304-3975(81)90110-9}

\bibitem{Sickert19}
Sickert, S.: A Unified Translation of Linear Temporal Logic to
  {\(\omega\)}-Automata. Ph.D. thesis, Technical University of Munich, Germany
  (2019),
  \url{https://nbn-resolving.org/urn:nbn:de:bvb:91-diss-20190801-1484932-1-4}

\bibitem{SickertE20}
Sickert, S., Esparza, J.: An efficient normalisation procedure for linear
  temporal logic and very weak alternating automata. In: {LICS}. pp. 831--844.
  {ACM} (2020). \doi{10.1145/3373718.3394743}

\bibitem{Zuck86}
Zuck, L.D.: {Past temporal logic}. Ph.D. thesis, The Weizmann Institute of
  Science, Israel (Aug 1986)

\end{thebibliography}

\appendix
\FloatBarrier
\newpage

%
%

\section{Appendix}

Let $\varphi \equiv^w \psi$ denote $w_k \models \varphi$ if{}f $w_k \models \psi$ for all $k \in \N$. The next two straightforward lemmas will be used pervasively in the following proofs.

\begin{lemma}
\label{lemma:nnf}
	For every formula $\varphi$ in negation normal form (and thus for every LTL formula we consider in this article), and for every two formulas $\psi$ and $\psi'$, $\psi \models \psi'$ implies $\varphi[\psi] \models \varphi[\psi']$.
\end{lemma}

\begin{lemma}
\label{lemma:replace}
	For every formula $\varphi$ and word $w$, $\psi \equiv^w \psi'$ implies $\varphi[\psi] \equiv^w \varphi[\psi']$.
\end{lemma}


\noindent\textbf{Lemma \ref{lem:UWandWU}.} 
(\ref{eqWU}) $\varphi_1 \W \varphi_2[\psi_1 \U \psi_2] \equiv (\varphi_1 \U \varphi_2[\psi_1 \U \psi_2]) \vee \G \varphi_1$, \newline
(\ref{eqUW}) $\varphi_1[\psi_1 \U \psi_2] \W \varphi_2 \equiv (\GF \! \psi_2 \wedge \varphi_1[\psi_1 \W \psi_2] \W \varphi_2) \vee \varphi_1[\psi_1 \U \psi_2] \U (\varphi_2 \vee (\G \varphi_1[\false]))$

\begin{proof}

For Equation (\ref{eqWU}) observe that, by the definition of the semantics of LTL, $\varphi_1 \W \varphi_2 \equiv  \varphi_1 \U \varphi_2 \vee \G \varphi_2$ holds for arbitrary formulas $\varphi_1, \varphi_2$. For Equation (\ref{eqUW}) we proceed as follows.

\smallskip\noindent ($\models$): Assume $w \models \varphi_1[\psi_1 \U \psi_2] \W \varphi_2$. If $w \models \GF \psi_2$, then we have $w_k \models \psi_1 \W \psi_2$ if{}f  $w_k \models \psi_1 \U \psi_2$ for every  $k \in \N$; by \Cref{lemma:replace} we then have $w \models \varphi_1[\psi_1 \W \psi_2] \W \varphi_2$, and we are done. If $w \not\models \GF \psi_2$, then there is $n \in \N$ such that $w_k \not\models \psi_2$ for all $k \geq n$, and so $\psi_1 \U \psi_2 \equiv^{w_n} \false$. Since $w$ satisfies the left-hand side, there are two possible cases
		\begin{itemize}
			\item There is $m \in \N$ such that $w_m \models \varphi_2$ and $w_k \models \varphi_1[\psi_1 \U \psi_2]$ for every $k < m$. Then, the second disjunct holds and we are done.
			\item $w_k \models \varphi_1[\psi_1 \U \psi_2]$ for every $k \in \N$. Then $w_k \models \varphi_1[\false]$ for all $k \geq n$ by \cref{lemma:replace} with $\psi_1 \U \psi_2 \equiv^{w_n} \false$. So $w_n \models \varphi_2 \vee \G \varphi_1[\false]$, and we are done.
		\end{itemize}

\noindent ($\leftmodels$): Assume $w$ satisfies the right-hand side formula. We consider two cases:
	\begin{itemize}
		\item $w \models \GF \psi_2 \wedge \varphi_1[\psi_1 \W \psi_2] \W \varphi_2$. Then $w \models \GF \psi_2$ and, as above,  we have $w_k \models \psi_1 \W \psi_2$ if{}f  $w_k \models \psi_1 \U \psi_2$ for every  $k \in \N$. Since $w \models \varphi_1[\psi_1 \W \psi_2] \W \varphi_2$, we get $w \models \varphi_1[\psi_1 \U \psi_2] \W \varphi_2$.
		\item $w \models \varphi_1[\psi_1 \U \psi_2] \U (\varphi_2 \vee \G \varphi_1[\false])$. Then there is $n \in \N$ such that $w_n \models \varphi_2 \vee \G \varphi_1[\false]$ and $w_k \models \varphi_1[\psi_1 \U \psi_2]$ for every $k < n$. We consider two cases: $w_n \models \varphi_2$ or $w_k \models \varphi_1[\false]$ for all $k \geq n$. In the first case, $w_n \models \varphi_1[\psi_1 \U \psi_2] \W \varphi_2$ by definition of $\W$. For the second, since $\false \models \psi_1 \U \psi_2$ and by \cref{lemma:nnf}, $w_k \models \varphi_1[\false]$ implies $w_k \models \varphi_1[\psi_1 \U \psi_2]$. Therefore, we have $w_k \models \varphi_1[\psi_1 \U \psi_2]$ for every $k < n$ and for every $k \geq n$, so $w \models \varphi_1[\psi_1 \U \psi_2] \W \varphi_2$.
\end{itemize}
\end{proof}

\noindent\textbf{Lemma \ref{lemma:GF1}.} (\ref{eqGF1}) $\varphi[\GF \psi] \equiv (\GF \psi \wedge \varphi[\true] ) \vee \varphi[\false]$, \newline
(\ref{eqFG1}) $\varphi[\FG \psi] \equiv (\FG \psi \wedge \varphi[\true] ) \vee \varphi[\false]$

\begin{proof}
We prove that $\varphi[\psi] \equiv (\psi \wedge \varphi[\true]) \vee \varphi[\false]$ for every $\psi$ such that $\G \psi \equiv \psi$, i.e., $w \models \psi$ if{}f $w_k \models \psi$ for all $k \in \N$.
This is a generalization of (3) and (4), since both $\GF \psi'$ and $\FG \psi'$ satisfy the requirement about $\psi$ for any $\psi'$. Since $\psi$ is satisfied by all suffixes or for no suffix of a word $w$, it follows that either $\psi \equiv^w \true$ or $\psi \equiv^w \false$ for any word $w$.

\medskip

\noindent ($\models$):  If $w$ satisfies $\psi$, we also have $w \models \varphi[\true]$ by \cref{lemma:replace}, and so the first disjunct is satisfied. Otherwise, again by \cref{lemma:replace}, we have $w \models \varphi[\false]$, so the second disjunct holds.

\medskip

\noindent ($\leftmodels$):  Suppose $w$ satisfies the first disjunct, then $w \models \psi$ and $\varphi[\psi] \equiv^w \varphi[\true]$ by \cref{lemma:replace}. Otherwise, the second disjunct holds, and then $w \models \varphi[\psi]$ by \cref{lemma:nnf} since $\false \models \psi$.
\end{proof}

\lemmaGFII*

\begin{proof}

\noindent Proof of Equation (\ref{eqGF2}).

\noindent ($\leftmodels$):  We prove the following claims, which immediately imply the result:
\begin{itemize}
\item $\GF \varphi [\psi_1 \U \psi_2]  \models  \GF \varphi [\psi_1 \W \psi_2]$. \\
Follows from $\psi_1 \U \psi_2 \models \psi_1 \W \psi_2$ and \cref{lemma:nnf}.
\item $\FG \psi_1 \wedge \GF \varphi [\true] \models \GF \varphi [\psi_1 \W \psi_2]$. \\
Assume $w \models \FG \psi_1 \wedge \GF \varphi [\true]$. On the one hand, there must be $n \in \N$ such that $w_k \models \G\psi_1$ for all $k \geq n$. Since $\G \psi_1 \models  \psi_1 \W \psi_2$, $w_k \vDash \psi_1 \W \psi_2$ for all $k \geq n$. On the other hand, $\varphi[\true]$ holds in infinitely many suffixes of $w$. By \cref{lemma:replace} and $\psi_1 \W \psi_2 \equiv^{w_k} \true$ for all $k \geq n$, $\varphi[\psi_1 \W \psi_2]$ holds infinitely often in $w$ and we are done.

\end{itemize}

\noindent ($\models$):  Assume $w$ satisfies $\GF \varphi [\psi_1 \W \psi_2]$, and so that  $w' \models \varphi [\psi_1 \W \psi_2]$ for infinitely many suffixes $w'$ of $w$. We prove the following three claims, which immediately imply the result:
\begin{itemize}
\item If $w \models \FG \psi_1$ then $w \models \GF \varphi [\true]$.\\
Since $w \models \FG \psi_1 \wedge \GF \varphi [\psi_1 \W \psi_2]$, infinitely many suffixes of $w$ satisfy $\G\psi_1$ and $\varphi [\psi_1 \W \psi_2]$. Since $\G\psi_1 \models \psi_1 \W \psi_2$, these suffixes also satisfy $\psi_1 \W \psi_2$, and so also $\varphi [\true]$.

\item If $w \not\models \FG \psi_1$, then $\psi_1 \U \psi_2 \equiv^w \psi_1 \W \psi_2$ because $\G \psi_1$ never holds and $\psi_1 \W \psi_2 \equiv \psi_1 \U \psi_2 \vee \G \psi_1$. Therefore, $w \models \GF[\psi_1 \W \psi_2]$ by \cref{lemma:replace} and the first clause of the disjunction.
\end{itemize}

\medskip 

\noindent Proof of Equation (\ref{eqFG2}).

\noindent ($\leftmodels$): If the second clause of the right-hand side disjunction is satisfied, $\FG \varphi[\psi_1 \U \psi_2]$ holds since $\varphi$ is in negation normal form and $\false \models \psi_1 \U \psi_2$. Otherwise, the first disjunct must be true, so for any word $w$ satisfying $\GF \psi_2$, $\psi_1 \W \psi_2 \equiv^w \psi_1 \U \psi_2$ and so they can be replaced inside the context by \cref{lemma:replace}.

\medskip \noindent ($\models$): Assume $w$ satisfies $\FG \varphi[\psi_1 \U \psi_2]$, i.e., there is an $n \in \N$ such that $w_k \models \varphi[\psi_1 \U \psi_2]$ for all $k \geq n$. We consider two cases whether $w \models \GF \psi_2$ or not.
\begin{itemize}
	\item If $w \models \GF \psi_2$, $\psi_1 \U \psi_2 \equiv^w \psi_1 \W \psi_2$ for every $k \in \N$, so $w_k \models \varphi[\psi_1 \W \psi_2]$ for all $k \geq n$ by \cref{lemma:replace}. Hence, the first disjunct holds.
	\item Otherwise, there is an $m \geq n$ such that $w_k \not\models \psi_2$ for all $k \geq m$. As a result, $w_k \not\models \psi_1 \U \psi_2$ and $\psi_1 \U \psi_2 \equiv^{w_k} \false$ for all $k \geq m$. Using that $w_k \models \varphi[\psi_1 \U \psi_2]$ and \cref{lemma:replace}, $w_k \models \varphi[\false]$ for all $k \geq m$, so $\FG \varphi[\false]$.
\end{itemize}

\end{proof}

\end{document}